\documentclass[a4paper,reqno]{amsart}
\usepackage[utf8]{inputenc}

\usepackage[english]{babel} 
\usepackage{amsmath,amssymb,amsfonts,amsthm}
\usepackage{mathrsfs, esint, dsfont}
\usepackage{moreverb,rotating,graphics,float}
\usepackage[colorlinks, linkcolor={blue!50!blue}, citecolor={red}]{hyperref}
\usepackage{framed,fancybox}
\usepackage{enumerate}
\usepackage{csquotes}
\usepackage{todonotes}
 \usepackage[foot]{amsaddr}
 
\usepackage[backend=bibtex, style=alphabetic, citestyle=alphabetic, bibencoding=utf8, maxbibnames=4]{biblatex}
\newtheorem{theorem}{Theorem}
\newtheorem{proposition}[theorem]{Proposition}
\newtheorem{remark}[theorem]{Remark}
\newtheorem{lemma}[theorem]{Lemma}

\newcommand\1{{\mathds 1}}

\def\C{{\mathbb C}}

\def\N{{\mathbb N}}
\def\NN{{\mathbb N}}

\def\R{{\mathbb R}}
\def\RR{{\mathbb R}}

\def\bA{{\mathbf A}}
\def\ba {{ \mathbf a}}
\def\bB{{\mathbf B}}
\def\bb{{\mathbf b}}

\def\bj{{\mathbf j}}
\def\bk{{\mathbf k}}

\def\bm{{\mathbf m}}
\def\bp{{\mathbf p}}

\def\bR{{\mathbf R}}

\def\bx{{\mathbf x}}

\def\by{{\mathbf y}}

\def\bsigma{{\mathbf \sigma}}

\def\bnull{{\mathbf 0}}

\def\rd{{\mathrm{d}}}
\def\re{{\mathrm{e}}}
\def\ri{{\mathrm{i}}}

\def\cB{{\mathcal B}}

\def\cE{{\mathcal E}}
\def\cF{{\mathcal F}}

\def\cH{{\mathcal H}}
\def\cI{{\mathcal I}}

\def\cP{{\mathcal P}}

\def\cR{{\mathcal R}}
\def\cS{{\mathcal S}}

\def\cZ{{\mathcal Z}}

\def\fS{{\mathfrak S}}

\newcommand{\Tr}{{\rm Tr} \,}
\newcommand{\VTr}{\underline{\rm Tr}}
\newcommand{\tr}{{\rm tr} \,}
\newcommand{\spinup}{\uparrow}
\newcommand{\spindown}{\downarrow}
\newcommand{\Lat}{\cR}
\newcommand{\RLat}{{\cR^*}}
\newcommand{\WS}{\Gamma}

\newcommand{\per}{{\rm per}}
\newcommand{\loc}{{\rm loc}}

\newcommand{\ie}{{\em i.e. }}
\renewcommand{\epsilon}{\varepsilon}

\newcommand{\LT}{{\rm LT}}

\newcommand{\norm}[1]{\left\| #1\right\|}
\newcommand{\set}[1]{\left\{ #1\right\}}
\newcommand{\bra}[1]{\left( #1\right)}
\newcommand{\av}[1]{\left| #1\right|}

\renewcommand{\phi}{\varphi}

\newcommand\uup{{\uparrow}}
\newcommand\down{{\downarrow}}
\newcommand\rHF{{\rm rHF}}

\let\Im\relax \DeclareMathOperator*\Im{Im}
\let\div\relax \DeclareMathOperator*\div{div}
\DeclareMathOperator*\curl{\bf curl}
\DeclareMathOperator*\Ker{Ker}
\DeclareMathOperator*\Ran{Ran}

\title[The rHF model with self-generated magnetic fields]{The reduced Hartree-Fock model with self-generated magnetic fields}

\author{David Gontier}
\address{CEREMADE, University of Paris-Dauphine, PSL University, 75016 Paris, France}

\author{Salma Lahbabi}
\address{ESSM, LRI, ENSEM, University Hassan II Casablanca, Route d'El Jadida, km 7, BP:8118, Oasis, Casablanca, Maroc}

\date{\today}

\begin{document}

\maketitle

\begin{abstract}
We study the well-posedness of the reduced Hartree-Fock model for molecules and perfect crystals when taking into account a self-generated magnetic field. We exhibit a critical value $\alpha_c > 0$ such that, if the fine structure constant $\alpha$ is smaller than $\alpha_c$, then the corresponding system is stable, whereas if $\alpha$ is greater than $\alpha_c$, it is unstable. We give an explicit characterisation of $\alpha_c$ as a minimisation problem over the set of zero-modes, and we prove that the critical values for the molecular case and the periodic case coincide. Finally, we prove the existence of minimisers when the system is neutral or positively charged.
\end{abstract}



\section{Introduction}

The reduced Hartree-Fock (rHF) model is a non-linear model introduced in~\cite{hartree1928wave,solovej1991proof} to describe molecular systems, as well as periodic perfect crystals~\cite{Catto2001,Catto1998_book}. This model is an approximation of the many-body Schrödinger model. It is a first step in a ladder of better approximations, but can still provide interesting physical features that allow to study infinite systems, such as crystals with local defects~\cite{Cances2008,CanEhr} and disordered systems~\cite{CaLaLe-12,La-13,BLBL2007,Blanc2003}.

In this work, we investigate the rHF model, when we include a self-generated magnetic field. This model was recently introduced in~\cite{carlos} for molecular systems. The present article aims at extending the results of~\cite{carlos} for both the molecular case, and the periodic one.

Self-generated magnetic fields were first studied in the Schrödinger model by Fröhlich, Lieb and Loss for the one-electron atom~\cite{FLL1}, and by Lieb and Loss for the many-electron atom and one-electron molecule~\cite{LL2} (see also~\cite{LY3} and more recent works~\cite{erdos2010ground, erdos2012scott, erdos2012relativistic, erdos2012second, erdos2013stability}).
In~\cite{FLL1}, the authors proved the existence of a critical value $\alpha_c > 0$ such that if the fine structure constant $\alpha$ is smaller than $\alpha_c$, then the corresponding system is stable while if $\alpha$ is greater than $\alpha_c$, then it is unstable (the energy is not bounded from below). This instability is caused by the so-called {\em zero modes}, which are non trivial pairs of functions $(\psi, \bA)$ solving
\[
    \bsigma \cdot \bra{-\ri\nabla+\bA}\psi= \bnull,
\]
where $\bsigma := (\sigma_x, \sigma_y, \sigma_z)$ are the usual Pauli matrices. The existence of such zero modes was first proved by Loss and Yau in~\cite{LY3}, and are now rather well-understood~\cite{erdos2001kernel}. 
It implies, using a scaling argument, that for $\alpha$ large enough, the energy is not bounded from below. In the case of one electron atoms~\cite{FLL1}, the authors characterised the critical value $\alpha_c$, above which instability occurs, as a minimisation problem over the set of zero-modes. No such result was given for other systems.


In the rHF model, the state of the electrons is described by a one body-density matrix $\gamma$, which can be seen as the projection operator on the occupied orbitals. For a system with $N$-electrons, $\gamma$ is usually a rank-$N$ projector. We therefore expect instability whenever $\gamma$ ``contains'' one  or more zero-modes. In this article, we prove results similar to~\cite{FLL1}, but where $\alpha_c$ is now characterised by a minimisation problem on a set of zeros modes $(\gamma, \bA)$ satisfying
\[
    \Ran \gamma \subset \Ker\bra{ \bsigma \cdot (-\ri \nabla + \bA)}, \quad \Tr \bra{\gamma} \le N.
\]
The inequality $\Tr (\gamma) \le N$ shows that only a fraction of $\gamma$ can lead to instability, a phenomenon which is classical in concentration-compactness arguments~\cite{Lions-84, Lions-84b}.


Our main result is that for a fixed maximal nuclear charge $z$, and a fixed number of electrons $N$ (the total number of electrons in the molecular case, and the number of electrons per unit cell in the periodic case), there is a critical value $\alpha_c(z,N) > 0$ such that for $\alpha <\alpha_c(z,N)$, the rHF problem is stable, in the sense that the energy is bounded from below, while for $\alpha > \alpha_c(z,N)$, the rHF energy is not bounded from below. We prove that the critical value $\alpha_c(z,N)$ is the same for molecules and for perfect crystals. This comes from the fact that instability is caused electrons that concentrate near one nucleus. By rescaling, we end up in both situations with the same functional to study.
In the case $\alpha < \alpha_c(z,N)$, we also prove that the problem is well-posed (\ie admits minimisers) whenever the system is neutral or positively charged. Our arguments follow the ones of~\cite{carlos}.


The article is organised as follows. In Section~\ref{sec:models_and_mainResults}, we recall the rHF models for finite and periodic systems when the self-generated magnetic field is included in the model, and we state our results. We prove the stability and instability results for the molecular case in Section~\ref{sec:proof_stablity}, and in the periodic case in Section~\ref{sec:proof:(in)stablity_per}. The properties of $\alpha_c(z,N)$ are studied in Section~\ref{sec:proof:betac}. Finally, in Section~\ref{sec:proof_existence}, we prove the existence of minimisers when $\alpha < \alpha_c(z,N)$ and the system is neutral or positively charged.

\subsubsection*{Acknowledgments.} The authors warmly thank Mathieu Lewin for stimulating discussion and help. This project has received funding from the ERC-MDFT (No 725528 of M.L.), and from PEPS-INSMI 2018 of D.G.


\section{Notation and main results}
\label{sec:models_and_mainResults}

We present in this section the rHF model with self-generated magnetic field, for both finite systems and periodic systems, and we state our results. 

\subsection{The magnetic rHF model for finite systems}

The rHF model with self-generated magnetic fields for molecular systems was  recently introduced  in~\cite{carlos}. It is a natural extension of the rHF model with no magnetic fields~\cite{hartree1928wave, solovej1991proof}. In these models, electronic systems are described by one-body density matrices
\[
    \gamma \in \cP := \left\{ \gamma \in \cS(L^2(\R^3, \C^2)), \ 0 \le \gamma \le 1 \right\},
\]
where $\cS(\cH)$ is the space of bounded self-adjoint operators acting on the Hilbert space $\cH$. Here, the Hilbert space $L^2(\R^3, \C^2)$ includes the spin degree of freedom. Such operators can be decomposed as a $2 \times 2$ matrix of the form $\gamma = \begin{pmatrix}
\gamma^{\spinup \spinup} & \gamma^{\spinup \spindown} \\ \gamma^{\spindown \spinup} & \gamma^{\spindown \spindown}
\end{pmatrix}.
$
For $N$ electrons, we have $\gamma \in \cP^N$, where
\[
\cP^N := \left\{ \gamma \in \cP, \ \Tr(\gamma) = N \right\}.
\]
For a state $\gamma \in \cP$, we denote by $\rho_\gamma(\bx)=\gamma^{\uup\uup}(\bx,\bx)+ \gamma^{\down\down}(\bx,\bx)$ its density, by $\bj_\gamma(\bx) := (\bp\gamma+\gamma\bp)(\bx, \bx)$ its current, and by $\bm_\gamma (\bx) := {\rm Tr}_{\C^2} \bra{\bsigma \cdot \gamma(\bx, \bx)}$ its magnetisation. 

To take into account magnetic fields, we follow~\cite{FLL1}, and introduce
\[
 	H^1_{\rm div}:= \left\{ \bA\in L^6(\RR^3, \RR^3),\ \bB := \curl \bA\in  L^2(\RR^3, \RR^3), \ \div \bA=0 \right\}.
\]
The condition $\div \bA=0$ is called the {\em Coulomb gauge}.  In this gauge, every $\bA\in H^1_{\div}$ satisfies $ \nabla \bA \in L^2(\R^3, \R^3)$, and
\[
 	\int_{\R^3} \bB^2 = \norm{\nabla \bA}_{L^2}^2.
\]
The total energy of $(\gamma, \bA) \in \cP^N \times H^1_{\div} $ in the rHF approximation is given by
\begin{equation} 
\label{eq:def:ErHF}
    \cE_\alpha(\gamma, \bA) := \frac12 \Tr \left( \left[ \bsigma \cdot ( \bp + \bA) \right]^2 \gamma \right) + \int_{\R^3} V \rho_{\gamma} + \frac12 D(\rho_\gamma, \rho_\gamma) + \frac{1}{8 \pi \alpha^2} \int_{\R^3} \bB^2.
\end{equation}
The first term is the Pauli kinetic energy in the presence of the magnetic field $\bA$. We denoted by $\bp := - \ri \nabla$ the momentum operator, and by $\bsigma = (\sigma_x, \sigma_y, \sigma_z)$ the Pauli matrices, defined by
\begin{equation*}
 \sigma_x=\begin{pmatrix}
0 & 1\\1 & 0
 \end{pmatrix}
    ,\quad 
\sigma_y=\begin{pmatrix}
0 & -\ri\\ \ri & 0
\end{pmatrix}
,\quad
\sigma_z=\begin{pmatrix}
1 & 0\\0 & -1
\end{pmatrix}.
\end{equation*}
Here and thereafter, $\Tr \left( \left[ \bsigma \cdot ( \bp + \bA) \right]^2 \gamma \right)$ is a short-hand notation for $$\Tr \left( \left[ \bsigma \cdot ( \bp + \bA) \right] \gamma \left[ \bsigma \cdot ( \bp + \bA) \right] \right).$$ The second term describes the interaction of the electrons with an external potential~$V$. In this article, we consider the case where $V$ is the Coulomb potential generated by a set of $M$ nuclei of charges $\{ z_j \}_{1 \le j \le M} \in (\R^+)^M$ located at {\em fixed} locations $\{ \bR_j \}_{1 \le j \le M} \in (\R^3)^M$, that is
\begin{equation} 
\label{eq:def:V}
V(\bx) := - \sum_{j=1}^M \frac{z_j}{| \bR_j-\bx |}.
\end{equation}
The third term is the Hartree energy, defined by the quadratic form
\[
D(\rho_1, \rho_2) := \iint_{(\R^3)^2}\dfrac{\rho_1(\bx) \rho_2(\by)}{| \bx - \by |} \rd \bx \rd \by.
\]
The last term is the energy of the magnetic field $\bB$. The constant $\alpha$ is the fine structure constant. In our system of unit where $e = 1$, $\hbar = 1$ and $m = 1$, we have $\alpha = e^2 / \hbar c = 1/c$, where $c$ is the speed of light. Although its true physical value is $\alpha \approx \frac{1}{137}$, we keep $\alpha$ as a variable, and study the models for different values of $\alpha$.

The rHF energy with self-generated magnetic field is defined as
\begin{equation} \label{eq:def:Ialpha}
 	  \boxed{ I(\alpha,N) := \inf \left\{ \cE_\alpha(\gamma, \bA), \ \gamma \in \cP^N, \ \bA \in H^1_{\div} \right\}.}
\end{equation}

\subsection{Main results in the finite systems case} 
\label{ssec:main-results_finiteCase}
As first proved by Fröhlich, Lieb, Loss and Yau in a series of papers~\cite{FLL1, LL2,LY3}, the stability of atoms and molecules with a self-generated magnetic field fails because of the existence of zero-modes. 
In the rHF model, we denote the set of zero-modes by
\[
\cZ := \left\{ (\gamma, \bA) \in  \cP \times H^1_{\div}, \quad \Ran \gamma \subset \Ker \left( \bsigma \cdot (\bp + \bA) \right)  \right\}.
\]
The critical value $\alpha_c(z,N)$, above which the energy is no longer bounded from below, is given by a minimisation problem over the set of zero-modes $\cZ$. We introduce
\begin{equation} \label{eq:def:betazN}
\boxed{ \beta(z,N) := \inf \left\{ \frac12 D(\rho_\gamma, \rho_\gamma) - z \int_{\R^3} \dfrac{\rho_\gamma(\bx)}{| \bx |}, \quad
    (\gamma, \bA) \in \cZ, \ \int_{\R^3} \bB^2 = 1, \ \Tr(\gamma) \le N \right\}.}
\end{equation}
The properties of $\beta(z,N)$ are given in Proposition~\ref{prop:betac} below. In particular, $\beta(z,N)$ is negative and we set 
\begin{equation} \label{eq:def:alphac}
\alpha_c(z, N) := \left( \frac{-1}{8 \pi \beta(z,N)} \right)^{1/2}.
\end{equation}

Our first main result shows that~\eqref{eq:def:Ialpha} is bounded from below if $\alpha < \alpha_c(z,N)$, and is not bounded from below if $\alpha > \alpha_c(z,N)$. The proof is presented in Section~\ref{sec:proof_stablity}. 

\begin{theorem}[Stability and instability in the finite case]
    \label{th:(in)stability}
    Let $V$ be a molecular Coulomb potential of the form~\eqref{eq:def:V}, and let $z := \max_{1 \le j \le M} \{ z_j \}$ be the maximal nuclear charge in the system.
    \begin{itemize}
        \item For all $0 \leq \alpha <  \alpha_{c}(z,N)$ the energy $\cE_\alpha$ defined in~\eqref{eq:def:ErHF} is bounded from below on $\cP^N \times H^1_{\div}$. The map $\alpha \mapsto I(\alpha^{-2},N)$ is continuous, concave and non-decreasing on $[0, \alpha_c) $. 
        \item For all $\alpha > \alpha_{c}(z,N)$, the energy $\cE_\alpha$ is not bounded from below on $\cP^N \times H^1_{\div}$.
    \end{itemize}
\end{theorem}

\begin{remark}
    The only negative term in~\eqref{eq:def:ErHF} is the term $\int V \rho$. If we consider smeared nuclei, where $V=\mu*\frac{1}{\av{x}} \in L^\infty$, then the energy is always bounded from below (by $\| V \|_\infty N$ for instance), for all value of $\alpha > 0$.
\end{remark}

When $\alpha < \alpha_c(z,N)$, the energy~\eqref{eq:def:ErHF} is bounded from below. We now examine whether this problem has minimisers. The proof of the following result can be found in Section~\ref{sec:proof_existence}. It follows the one of~\cite{carlos}. 

\begin{theorem}[Existence of minimisers in the finite case]
    \label{th:existence_minimiser}
    If $\alpha <\alpha_c(z,N)$ and $N \le Z := \sum_{j=1}^M z_j$, then the problem~\eqref{eq:def:Ialpha} admits a minimiser. If $(\gamma_\alpha, \bA_\alpha) \in \cP^N \times H^1_{\div}$ is such a minimiser, then it satisfies the Euler-Lagrange equations
    \begin{equation} \label{eq:Euler-Lagrange}
    \begin{cases}
    \displaystyle \gamma_\alpha = \1\bra{ H_{\rho_\alpha, \bA_\alpha}\leq \varepsilon_F}+\delta, \\
    \displaystyle
    H_{\rho_\alpha, \bA_\alpha} := \frac{1}{2}\left[ \bsigma \cdot (\bp+\bA_\alpha) \right]^2 + V + \rho_\alpha * | \bx |^{-1}, \\
    \displaystyle \frac12 \left( \bj_\alpha + \curl \bm_\alpha \right) + \bA_\alpha \rho_\alpha + \dfrac{1}{4 \pi \alpha^2}(-\Delta \bA_\alpha) = \bnull,
    \end{cases}
    \end{equation}
    where $\varepsilon_F \in \R$, called the {\em Fermi energy}, is chosen so that $\Tr(\gamma_\alpha) = N$, and $\delta$ is an operator satisfying $0\leq \delta \leq \1 \bra{ H_{\rho_\alpha, \bA_\alpha}= \varepsilon_F}$. \\
    In addition, $\Ran \gamma_\alpha \subset H^2(\R^3, \C^2)$, and $\bA_\alpha \in W^{2,6}(\R^3, \R^3)$.
\end{theorem}

\begin{remark}
    The quantity $\bj + \bA \rho$, sometimes called the {\em physical current}, is gauge invariant. If $\bA' = \bA + \nabla \mu$ and $\gamma' = \re^{ \ri \mu(\cdot)} \gamma \re^{ - \ri \mu(\cdot)}$, then $\bj' + \bA' \rho = \bj + \bA \rho$. Taking the divergence in the last equation gives the {\em continuity equation}  $\div \left( \bj + \bA \rho  \right) = 0$.
\end{remark}


We end this section with some properties of the function $\beta(z,N)$ defined in~\eqref{eq:def:betazN}. The proof of the following Proposition can be read in Section~\ref{sec:proof:betac}.
\begin{proposition} \label{prop:betac}
    ~ 
    \begin{enumerate}[(i)]
        \item For all $z > 0$ and for all $N > 0$, $\beta(z,N) < 0$.
        \item The maps $N \mapsto \beta(z,N)$ and $z \mapsto \beta(z,N)$ are non-increasing.
        \item There is $C_\beta > 0$ such that the limit $\displaystyle \beta_c(z) := \inf_{N \in \N^*} \beta(z,N) = \lim_{N \to \infty} \beta(z,N)$ satisfies
        $$\beta_c(z) \ge -C_\beta z^{7/6} \qquad (> -\infty).$$ 
        \item For any $N$ and $z$, the problem~\eqref{eq:def:betazN} defining $\beta(z,N)$ admits a minimiser. 
    \end{enumerate}
\end{proposition}

The third point implies that if $\alpha^2 z^{7/6} < (8 \pi C_\beta)^{-1}$, then the system is stable for all $N$. This inequality was already proved in~\cite{LL2} for the one-electron atom in the Schrödinger model.

\subsection{The magnetic rHF model for periodic systems}
\label{sec:presentation_per}
The rHF model can describe molecular systems, as well as infinite periodic systems. The energy (per unit cell) of such crystals is defined by mean of thermodynamic limit in~\cite{Catto2001}. The authors proved that the limit has a simple characterisation, and the resulting model was extensively studied~\cite{Cances2008, cances2010dielectric, Cances2012,GL2015,Gontier2016supercell}.

Let $\Lat$ be a periodic lattice representing the periodicity of the crystal. Let $\WS$ be the Wigner-Seitz cell of this lattice, and let $\RLat$ be the reciprocal lattice. The electrons interact with a periodic arrangement of $M$ nuclei of charges $\{ z_j\}_{1 \le j \le M} \in (\R^+)^M$ located at $\{ \bR_j \}_{1 \le j \le M} \in (\WS)^M$. The periodic nuclear potential is similar to the molecular one in~\eqref{eq:def:V}, and is given by
\begin{equation}
 \label{eq:def:V_per}
         V_\per(\bx) := - \sum_{j=1}^M z_j G_\Lat(\bx - \bR_j),
\end{equation}
where $G_\Lat$ is the Green's function of the $\Lat$-periodic Laplace operator, solution of
$$
-\Delta G_\Lat = 4\pi \bra{-1 +\sum _{\bR \in \Lat }\delta_\bR},
\quad \text{or} \quad 
G_\Lat(\bx) := 4\pi \sum_{\bk \in \RLat \setminus \{ \bnull \} }\frac{\re^{\ri \bk\cdot \bx}}{|\bk|^2}.
$$
The total charge of the nuclei in a unit cell is $Z := \sum_{j=1}^M z_j$. For the crystal to be neutral, we need to take $N = Z$ electrons per unit cell, so in the periodic case, the number $N$ is fixed.

\medskip

An $\Lat$-periodic electronic system is described by a periodic one-body density matrix
\[
    \gamma \in \cP_\per^N := \left\{ \gamma\in \cS\bra{L^2(\RR^3,\C^2)},\ 0\leq \gamma \leq 1,\ \forall \bk\in \Lat, \ \tau_\bk\gamma=\gamma\tau_\bk, \ \VTr_\Lat (\gamma) = N \right\},
\]
where $\tau_\bk:f\in L^2(\RR^3)\mapsto f(\cdot+\bk)$ is the usual translation operator, and $\VTr_\Lat$ denotes the trace per unit cell, defined for any locally trace class operators that commute with $\Lat$-translations by $\VTr_\Lat(\gamma)= \Tr \bra{\1_\WS \gamma \1_\WS}$.

Concerning magnetic fields, it is natural to consider $\Lat$-periodic magnetic fields $\bB$. If $\bA$ is periodic, then so is $\bB$, but the converse is not true: a constant magnetic field $\bB$ has a diverging field $\bA$. In this article, we focus on the easier case where $\bA$ is also periodic. It would be interesting to extend our results without this extra assumption. However, if $\bA$ is not periodic, the operator $[ \bsigma \cdot( \bp + \bA)]^2 \gamma$ no longer commutes with $\Lat$-translations, which creates some difficulties. 

For any constant vector $\ba_0 \in \R^3$, the field $\bA + \ba_0$ produces the same $\bB$. We therefore consider
$$
H^1_{\div, \per} := \left\{ \bA\in H^1_\per(\R^3, \R^3), \ \bB := \curl \bA\in L^2_\per(\R^3, \R^3),\  \div \bA=0, \ \int_\WS \bA= \bnull \right\}.
$$
The periodic rHF energy per unit cell of $(\gamma, \bA) \in \cP^N_\per \times  H^1_{{\div}, \per}$ is
\begin{equation} \label{eq:def:ErHF_per}
    \cE_{\per, \alpha}(\gamma,\bA) :=
    \frac{1}{2}\VTr_\Lat \bra{ [ \bsigma \cdot (\bp + \bA)]^2 \gamma}
    +\int_\WS V_\per \rho_\gamma+\frac{1}{2}D_\Lat(\rho_\gamma, \rho_\gamma)+ \frac{1}{8\pi \alpha^2}\int_\Gamma \bB^2.
\end{equation}
Here, the periodic Hartree term is defined with the quadratic form
$$
D_\Lat(f,g) := \iint_{(\WS)^2} G_\Lat(\bx-\by)f(\bx)g(\by) \rd \bx \rd \by.
$$
Finally, the rHF ground state energy with self-generated magnetic field per unit cell is given by
\begin{equation} \label{eq:def:Ialpha_per}
    \boxed{I_\per(\alpha) := \inf \left\{ \cE_{\per,\alpha}(\gamma,\bA), \ \gamma \in \cP_\per^N, \ \bA \in H^1_{\div, \per} \right\}.}
\end{equation}


\subsection{Main results in the periodic case}

Our result in the periodic setting are similar to the ones in the molecular case. First, we have the following Theorem, which is the periodic equivalent of Theorem~\ref{th:(in)stability}.

\begin{theorem}[Stability and instability in the periodic case]
    \label{th:(in)stability_per}
    Let $V_\per$ be a periodic nuclear potential of the form~\eqref{eq:def:V_per}, let $z := \max_{1 \le j \le M} \{ z_j \}$ be the maximal nuclear charge in the system, and let $Z := \sum_{j=1}^M z_j$ be the total charge per unit cell. Then, for  $N = Z$,
    \begin{itemize}
        \item For all $0 \leq \alpha <  \alpha_{c}(z,N)$ the energy $\cE_{\per, \alpha}$ defined in~\eqref{eq:def:ErHF_per} is bounded from below on $\cP^N_\per \times H^1_{\div, \per}$. The map $\alpha \mapsto I_\per(\alpha^{-2})$ is continuous, concave and non-decreasing on $[0, \alpha_c) $.
        \item For all $\alpha > \alpha_{c}(z,N)$, the energy $\cE_{\per, \alpha}$ is not bounded from below on $\cP^N_\per \times H^1_{\div, \per} $.
    \end{itemize}
\end{theorem}

\begin{remark}
    The critical value $\alpha(z,N)$ is the same for the finite system and for the periodic one. This is because, when $\alpha > \alpha_c(z,N)$, minimising sequences for the two problems both concentrate near the nucleus with highest charge $z$. After rescaling, we obtain the definition of $\beta(z,N)$ in~\eqref{eq:def:betazN}. It involves one-body density matrices with $N$ electrons or less: some electrons may not participate in the creation of the unstable zero-mode.
\end{remark}

Since the system is always neutral in the periodic case, minimisers always exist.

\begin{theorem}[Existence of minimisers in the periodic case]
    \label{th:existence_minimiser_per}
    If $\alpha <\alpha_c(z,N)$ and $Z := \sum_{j=1}^M z_j = N$, then the problem~\eqref{eq:def:Ialpha_per} admits a minimiser. In addition, if $(\gamma_\alpha, \bA_\alpha) \in \cP^N_\per \times H^1_{\div, \per}$ is a minimiser, then it satisfies the Euler-Lagrange equations
    \begin{equation} \label{eq:Euler-Lagrange_per}
    \begin{cases}
    \displaystyle \gamma_\alpha = \1\bra{ H_{\rho_\alpha, \bA_\alpha}^\per \leq \varepsilon_F}, \\
    \displaystyle H_{\rho_\alpha, \bA_\alpha}^\per := \frac{1}{2}\left[ \bsigma \cdot (\bp+\bA_\alpha) \right]^2 + V_\per + \rho_\alpha * G_\Lat, \\
    \displaystyle \frac12 \left( \bj_\alpha + \curl \bm_\alpha \right) + \bA_\alpha \rho_\alpha + \dfrac{1}{4 \pi \alpha^2} (- \Delta \bA_\alpha) =
    \mu.
    \end{cases}
    \end{equation}
    where $\varepsilon_F \in \R$, called the {\em Fermi energy}, is chosen so that $\VTr_\Lat(\gamma_\alpha) = N$, and $\mu \in \R$ is chosen so that $\int_\WS \bA_\alpha = \bnull$. \\
    In addition, $\Ran \gamma_\alpha \subset H^2_\per(\R^3, \C^2)$, and $\bA_\alpha \in W^{2,6}_\per(\R^3, \R^3)$.
\end{theorem}

In our definition of the periodic model~\eqref{eq:def:ErHF_per}, we have constrained $\gamma$ and $\bA$ to commute with $\Lat$-translations (in particular, the density $\rho_\gamma$ is $\Lat$-periodic). We would like to prove the thermodynamic limit, that is to study the model when we restrict $\gamma$ to commute with $L \Lat$-translations, and take the limit $L \to \infty$. When $\alpha = 0$ (no magnetic field), the problem is strictly convex in $\rho_\gamma$, hence the $L \Lat$-periodic problem is equivalent to the $\Lat$-periodic problem. In the general case however, we may have symmetry breaking.

One simple case where symmetry breaking might happen is the following. Assume there is $z \le N_1 < N_2$ such that $\beta(z, N_1) > \beta(z, N_2)$. Consider the $\Lat$-periodic problem with an arrangement of nuclei with highest charge $z$ and total charge $Z = N_1$, and let $\alpha_2 < \alpha < \alpha_1$, where $\alpha_i := (-8 \pi \beta_c(z, N_i))^{-1/2}$. Since $\alpha < \alpha_1$, the $\Lat$-periodic problem is bounded from below: $N_1$ electrons are not sufficient to create instability. However, for $L$ large enough so that $L^3 N_1 > N_2$, the $L \Lat$-periodic problem has at least $N_2$ electrons which can break symmetry and gather near one of the atom with charge $z$. Since $\alpha > \alpha_2$, the corresponding energy is not bounded from below. We do not know whether this case can happen, that is if there is $z \le N_1 < N_2$ with $\beta(z, N_1) > \beta(z, N_2)$.

The rest of the paper is devoted to the proof of Theorems~\ref{th:(in)stability},\ref{th:(in)stability_per} and Theorems~\ref{th:existence_minimiser},\ref{th:existence_minimiser_per}.

\section{Proof of Theorem~\ref{th:(in)stability}: instability for finite systems}
\label{sec:proof_stablity}

Our arguments follow the ones of~\cite{FLL1}. The main difference is that we need to deal with one-body density matrices $\gamma$ instead of wave-functions, so our model is non-linear. We break the proof into several steps for clarity.

\subsection*{Step 1. First properties of $I(\alpha)$.}
The function $I(\alpha)$ defined in~\eqref{eq:def:Ialpha} is the minimum of functions $\cE_\alpha$ that are linear non-decreasing in $\alpha^{-2}$. This shows that $\alpha \mapsto I(\alpha^{-2})$ is concave and non-decreasing, hence it is continuous on its domain. In particular, $\alpha \mapsto I(\alpha)$ is also continuous and non-increasing on its domain, and we can define
\[
    \widetilde{\alpha_c} := \inf \left\{ \alpha \ge 0, \ I(\alpha) = -\infty  \right\},
\]
so that $I(\alpha) > -\infty$ if $ \alpha < \widetilde{\alpha_c}$, while $I(\alpha) = - \infty$ if $\alpha > \widetilde{\alpha_c}$. We now prove that $\widetilde{\alpha_c} = \alpha_c(z,N)$, where $\alpha_c(z,N)$ has been defined in~\eqref{eq:def:betazN}-\eqref{eq:def:alphac}. 

\subsection*{Step 2. First inequality: $\alpha_c(z,N) \ge \widetilde{\alpha_c}$.} We consider $\alpha > \alpha_c(z,N)$, and we prove that $\alpha \ge \widetilde{\alpha_c}$. By definition of $\beta(z,N)$ introduced in~\eqref{eq:def:betazN}, and since $\alpha > \alpha_c(z,N)$, there exists a zero-mode $(\gamma, \bA) \in \cZ$ with $\int_{\R^3} \bB^2 = 1$ and $\Tr(\gamma) \le N$ such that 
\begin{equation} \label{eq:ineqAlpha}
    \dfrac{-1}{8 \pi \alpha^2} > \frac12 D(\rho_{\gamma}, \rho_{\gamma}) - z \int_{\R^3} \dfrac{\rho_{\gamma}(\bx)}{| \bx |}.
\end{equation}
The idea is to concentrate this zero mode near the nucleus of highest charge $z$. Without loss of generality, we may assume that $z_1 = z$ and $\bR_1 = \bnull$. We set, for $\lambda > 0$,
\begin{equation} \label{eq:dilation}
\gamma_\lambda' (\bx, \by) := \lambda^3 \gamma(\lambda \bx, \lambda \by), \qquad \bA_\lambda'(\bx, \by)  := \lambda \bA(\lambda \bx).
\end{equation}
We have $\Tr(\gamma_\lambda') = \Tr(\gamma) \le N$.
The kinetic energy of $\gamma_\lambda'$ is given by
\begin{equation}\label{eq:scaling1}
    \Tr \left( \left[ \bsigma \cdot ( \bp + \bA_\lambda') \right]^2 \gamma_\lambda' \right)  = \lambda^2 \Tr \left( \left[ \bsigma \cdot ( \bp + \bA) \right]^2 \gamma \right),
\end{equation}
while, for the contribution of the first nucleus, the Hartree term, and the magnetic energy, we have
\begin{align} \label{eq:scaling2}
    & - z \int_{\R^3} \dfrac{\rho_\lambda'(\bx)}{| \bx |} \rd \bx + \frac12 D(\rho_\lambda', \rho_\lambda') + \dfrac{1}{8 \pi \alpha^2} \int_{\R^3} (\bB_\lambda')^2 \\
    & \qquad  = \lambda \left( - z \int_{\R^3} \dfrac{\rho(\bx)}{| \bx |} \rd \bx + \frac12 D(\rho, \rho) + \dfrac{1}{8 \pi \alpha^2} \int_{\R^3} \bB^2 \right). \nonumber
\end{align}

For a general state, the limit $\lambda \to \infty$ would explode because of the $\lambda^2$ scaling in the kinetic energy. However, since $(\gamma, \bA)$ is a zero-mode, the kinetic energy vanishes. On the other hand, from~\eqref{eq:ineqAlpha} and the fact that $\int_{\R^3} \bB^2 = 1$, the parenthesis in the right-hand side of~\eqref{eq:scaling2} is  negative. This would prove that $\cE_\alpha(\gamma_\lambda', \bA_\lambda') \to -\infty$ as $\lambda \to \infty$. 

It remains to control the contribution of the other nuclei. In order to have a unified proof that also works in the periodic case, we present a proof based on cut-off functions. We believe that simpler methods are possible in the molecular case. 
 
Let $r > 0$ be the minimal distance between two nuclei, and let $\chi$ and $\chi'$ be radially decreasing cut-off functions such that 
\begin{equation} \label{eq:def:chichi'}
    \1 (\cB(\bnull, r/4)) \le \chi \le \1 \left( \cB(\bnull, r/2) \right),
    \ \text{and} \
    \1 (\cB(\bnull, r/2)) \le \chi' \le \1 \left(\cB(\bnull, 3r/4) \right).
\end{equation}
We set
\[
    \gamma_\lambda := \chi \gamma_\lambda' \chi \quad \text{and} \quad  \bA_\lambda := \chi' \bA_\lambda'.
\]
By construction, we have $\gamma_\lambda \in \cP$ and $\Tr(\gamma_\lambda) \le \Tr(\gamma'_\lambda) \le N$. Also, since $\chi \chi' = \chi$, we have that $\bA_\lambda = \bA_\lambda'$ on the support of $\chi$. We compute $\cE(\gamma_\lambda, \bA_\lambda)$, and prove that it converges to $-\infty$ as $\lambda \to \infty$. 

\medskip

\noindent \underline{Bound on the kinetic energy.} 
%
Since $\chi \chi' = \chi$, we have $\chi\bA_\lambda = \chi \chi' \bA_\lambda' = \chi \bA_\lambda'$, hence
\begin{align*}
     \bsigma \cdot ( \bp + \bA_\lambda) \chi  
     & = \bsigma \cdot ( \bp + \bA_\lambda') \chi  
      = \left[ \bsigma \cdot ( \bp + \bA_\lambda'), \chi \right] + \chi  \bsigma \cdot ( \bp + \bA_\lambda') \\
     & = -\ri \bsigma \cdot \nabla \chi + \chi \bsigma \cdot (\bp + \bA_\lambda').
\end{align*}
Together with the fact that $(\gamma_\lambda', \bA_\lambda')$ is a zero-mode, we get
\begin{align*}
    & \Tr \left(  \left[\bsigma \cdot ( \bp + \bA_\lambda) \right]^2 \chi \gamma_\lambda' \chi \right)  =  \Tr \left( 
    \left(\bsigma \cdot ( \bp + \bA_\lambda) \chi  \right)
    \gamma_\lambda' 
    \left(\chi \bsigma \cdot ( \bp + \bA_\lambda) \right)
    \right) \\
    & \quad = \Tr \left( (\bsigma \cdot \nabla \chi)^2 \gamma_\lambda' \right)
   = \int_{\R^3} \rho_\lambda' | \nabla \chi |^2 = \lambda \int_{\R^3} \rho(\by) \left( \lambda^{-1} | \nabla \chi |^2(\lambda^{-1} \by ) \right) \rd \by.
\end{align*}
The sequence $\lambda^{-1} | \nabla \chi |^2(\lambda^{-1} \by)$ is bounded in $L^\infty(\R^3)$ and in $L^3(\R^3)$, and has a support which goes to infinity as $\lambda \to +\infty$, hence it converges weakly(-$\ast$) to $0$ in $L^3(\R^3) \cap L^\infty(\R^3)$ as $\lambda \to +\infty$.
On the other hand, $\rho \in L^1(\R^3)$. Actually, we prove later (see Remark~\ref{rem:rho_zeromode} below) that for any zero-mode $(\gamma, \bA) \in \cZ$, we have $\rho_\gamma \in L^1(\R^3) \cap L^{3}(\R^3)$. Hence the integral goes to $0$ as $\lambda \to +\infty$. This proves that the kinetic energy behaves as $o(\lambda)$.

\medskip

\noindent \underline{Bound on the potential energy.} We now turn to the potential energy, and prove that the leading order is $- \lambda z \int_{\R^3} | \bx |^{-1} \rho(\bx) \rd \bx$. 
We separate the first nucleus from the others. For $j >1$, the function $\chi^2(\bx) | \bx - \bR_j |^{-1}$ is bounded. Hence, the contribution of the other nuclei is uniformly bounded with
\[
    \left| -\sum_{j=2}^M z_j \int_{\R^3} \dfrac{\rho_\lambda'(\bx) \chi^2(\bx)}{| \bx - \bR_j |} \rd \bx \right| 
    \le C \sum_{j=2}^M z_j \int_{\R^3} | \rho_\lambda'(\bx) | \rd \bx
    \le C (Z-z) N.
\]
For the first nucleus, we have
\begin{align*}
    \left| \lambda \int_{\R^3} \dfrac{\rho(\bx)}{| \bx |} \rd \bx - \int_{\R^3} \dfrac{\chi^2(\bx) \rho_\lambda'(\bx)}{| \bx |} \rd \bx  \right| & = \lambda \int_{\R^3} \rho(\bx) \dfrac{ \left[1 - \chi^2(\lambda^{-1}\bx) \right] }{| \bx |}  \rd \bx \\
    & \le \lambda N \sup_{\bx \in \R^3} \left( \dfrac{1 - \chi^2(\lambda^{-1}\bx)}{| \bx |} \right).
\end{align*}
Since $1 - \chi^2(\lambda^{-1} \bx)$ vanishes for $\av {\bx }< \lambda r/4$, the last supremum is bounded by $4 (\lambda r)^{-1}$. This proves that the difference is uniformly bounded by $4 N r^{-1}$ for all $\lambda > 0$.

\medskip

\noindent \underline{Bound for the Hartree term.} Again, we prove that the leading order is $\lambda D(\rho, \rho)$. 

We compute the difference
\begin{align*}\
    \left| \lambda D(\rho, \rho) - D(\chi^2 \rho_\lambda', \chi^2 \rho_\lambda') \right| = \lambda \iint_{\R^3} \dfrac{\rho(\bx) \rho(\by)}{| \bx - \by|} \left[ 1 - \chi( \lambda^{-1} \bx) \chi(\lambda^{-1} \by) \right] \rd \bx \rd \by.
\end{align*}
By dominated convergence, the last integral goes to $0$ as $\lambda \to \infty$, hence the difference is $o(\lambda)$.

\medskip

\noindent \underline{Bound on the magnetic energy.} We finally bound $\| \bB_\lambda \|_{L^2}^2$, where $\bB_\lambda := \curl \bA_\lambda$. Since $\bA_\lambda(\bx) = \lambda \chi'(\bx) \bA(\lambda \bx)$, we have
\[
    \bB_\lambda(\bx) = \underbrace{\lambda \nabla \chi'(\bx) \wedge \bA(\lambda \bx)}_{\ba_\lambda(\bx)} + \underbrace{\lambda^2 \chi'(\bx) \bB(\lambda \bx)}_{\bb_\lambda(\bx)}.
\]
Therefore 
$$
    \int_{\R^3}\bB_\lambda^2  -\lambda \int_{\RR^3} \bB^2= \int_{\RR^3} \ba_\lambda^2+2\text{Re}\bra{ \int_{\RR^3}\ba_\lambda \bb_\lambda}+ \int_{\RR^3} \bb_\lambda^2 -\lambda \int_{\RR^3} \bB^2,
$$
For the the first term, we have
$$
\int_{\RR^3} \ba_\lambda^2 = \lambda \int_{\R^3} | \bA |^2(\by) \left(\lambda^{-2} | \nabla \chi'(\lambda^{-1} \by) |^2 \right)\rd \by.
$$
The sequence  $\lambda^{-2} |\nabla \chi'(\lambda^{-1} \by) |^2$ is bounded in $L^{3/2}(\R^3)$, and the support goes to infinity as $\lambda \to \infty$, hence it converges weakly to $0$ in $L^{3/2}(\R^3)$. On the other hand, we have $\bA \in L^6(\R^3)$, hence $\bA^2 \in L^3(\R^3)$. We deduce that this term is $o(\lambda)$. Similarly, we have
$$
\int_{\RR^3} \bb_\lambda^2 -\lambda \int_{\RR^3} \bB^2= \lambda \int_{\R^3}\bra{ | \chi'(\lambda^{-1}\by) |^2 -1}| \bB |^2(\by) \rd\by= o(\lambda),
$$
where we used the fact that $\bB \in L^2(\R^3)$ and dominated convergence. Finally, by Cauchy-Schwarz inequality,
\begin{align*}
\av{\int_{\RR^3}\ba_\lambda \bb_\lambda}& \leq \bra{\int_{\RR^3} a_\lambda^2}^{1/2}\bra{\int_{\RR^3} b_\lambda^2}^{1/2}=o\bra{\sqrt{\lambda}}O\bra{\sqrt{\lambda}}=o(\lambda). 
\end{align*}

\medskip

\noindent \underline{Conclusion}
Altogether, we proved that,
\[
    \cE^\rHF_\alpha(\gamma_\lambda, \bA_\lambda) = \lambda \left(    - z \int_{\R^3} \dfrac{\rho(\bx)}{| \bx |} \rd \bx + \frac12 D(\rho, \rho) + \dfrac{1}{8 \pi \alpha^2} \int_{\R^3} \bB^2    \right) + o(\lambda).
\]
As we already mentioned, the constant in parenthesis is negative, hence this energy goes to $-\infty$ as $\lambda \to \infty$. In addition, for all $\lambda > 0$, $\Tr(\gamma_\lambda) \le N$. By adding some {\em missing} electrons at infinity, we deduce that $\cI(\alpha, N) = - \infty$. This proves that $\alpha \ge \widetilde{\alpha_c}(z,N)$, as claimed.

\subsection*{Step 3. Second inequality: $\widetilde{\alpha_c} \ge \alpha_c(z,N)$.}
Let $\alpha > \widetilde{\alpha_c}$, and let us prove that $\alpha \ge \alpha_c(z,N)$. By definition of $\widetilde{\alpha_c}$, there exists a sequence $(\widetilde{\gamma_n}, \widetilde{\bA_n}) \in \cP^N \times H^1_{\div}$ such that $\cE_\alpha(\widetilde{\gamma_n}, \widetilde{\bA_n}) \to - \infty$. We assume without loss of generality that $z = z_1$ with $\bR_1 = \bnull$.

\medskip

\noindent \underline{Localisation} We first localise our minimising sequences. Again, this allows to adapt the proof for the periodic setting. Recall that $\chi$ and $\chi'$ have been defined in~\eqref{eq:def:chichi'}. We set, for $1 \le j \le M$,
\[
    \chi_j(\bx) := \chi(\bx - \bR_j), \quad \text{and} \quad
    \chi_0(\bx) := \left( 1 - \sum_{j=1}^M \chi_j^2(\bx )  \right)^{1/2},
\]
so that $\sum_{j=0}^M \chi_j^2 = 1$, and
\[
   \chi_j'(\bx) := \chi'(\bx - \bR_j), \quad \text{and} \quad
        \chi_0'(\bx) := 1 - \sum_{j=1}^M \chi_j'(\bx ),
\] 
so that $\sum_{j=0}^M \chi_j' = 1$ (without the square). These functions are chosen so that for all $1 \le i \neq j \le M$, we have $\chi_i \chi_j = \chi_i' \chi_j' = \chi_i \chi_j' = 0$ and for all $1 \le i \le M$, $(1 - \chi_i') \chi_i = 0$. We use the functions $\chi_j$ to localise the one-body density matrices and the functions $\chi_j'$ to localise the magnetic fields.

The following Lemma gives an estimation of the localised energies.
\begin{lemma}\label{lem:localisation}
For any state $(\gamma,\bA) \in \cP^N\times H^1_{\div}$, we denote by
\begin{equation*} 
\gamma_j := \chi_j \gamma \chi_j\quad \text{and}\quad \bA_j=\chi_j'\bA.
\end{equation*}
Then we have $(\gamma_j,\bA_j)\in \cP\times H^1_{\div}$ with $\Tr\bra{\gamma_j}\leq N$. In addition, there is a constant $C \in \R^+$ such that, for any $(\gamma, \bA) \in \cP^N \times H^1_{\div}$ and any $\alpha > 0$, we have
\begin{align*}
    & \sum_{j=1}^M \left( \Tr \left( [\bsigma \cdot (\bp + \bA_j)]^2  \gamma_j  \right)
    + D( \rho_j , \rho_j )
    - z_j \int_{\R^3} \dfrac{\rho_j(\bx)}{| \bx - \bR_j |} \rd \bx \right) + \dfrac{1}{8 \pi \alpha^2} \int_{\R^3} \bB^2
    \\
   & \qquad \le \cE_\alpha(\gamma, \bA) + C.
\end{align*}

\end{lemma}

\begin{proof}[Proof of Lemma~\ref{lem:localisation}] Since $0 \le \gamma \le 1$, we have $0 \le \gamma_{j} \le 1$ and $\Tr(\gamma_j) \le \Tr(\gamma) = N$. Also, $\rho_j = \rho \chi_j^2$ and $\rho = \sum_{0}^M \rho_j$. For the kinetic energy, we use an IMS like formula (see~\cite[Theorem 3.2]{cycon})
    \[
    [\bsigma \cdot (\bp + \bA)]^2 = \sum_{j=0}^M  \chi_j \left[  \bsigma \cdot (\bp + \bA)\right]^2  \chi_j - \sum_{j=0}^M | \nabla \chi_j |^2.
    \]
    We deduce that
    \begin{align*}
    \sum_{j=0}^M \Tr \left( [\bsigma \cdot (\bp + \bA)]^2  \gamma_j  \right) & = \Tr \left( [\bsigma \cdot (\bp + \bA)]^2 \gamma \right) + \int_{\R^3} \left( \sum_{j=0}^m | \nabla \chi_j |^2 \right) \rho \\
    & \le \Tr \left( [\bsigma \cdot (\bp + \bA)]^2 \gamma \right) + CN,
    \end{align*}
    where we used the fact that the functions $| \nabla \chi_j |$ are bounded and $\int_{\R^3} \rho = N$ for the last inequality.
    We now drop the positive $j=0$ term on the left, and we localise the $\bA$ field. Since $(1 - \chi_j')\chi_j = 0$ for all $1 \le j \le M$, we have 
    $$ \Tr \left( [\bsigma \cdot (\bp + \bA)]^2 \gamma_j \right) =  \Tr \left( [\bsigma \cdot (\bp + \bA_j)]^2 \gamma_j \right). $$
    Altogether, we get
    \[
    \sum_{j=1}^M \Tr \left( [\bsigma \cdot (\bp + \bA_j)]^2 \gamma_j \right) \le \Tr \left( [\bsigma \cdot (\bp + \bA)]^2  \gamma  \right) + C.
    \]
    We now consider the Hartree term. Since $\rho = \sum_{j=0}^M \rho_j$ together with the fact that $D(f,g)$ is positive for $f, g \ge 0$, we directly get
    \[
    \sum_{j=1}^M D(  \rho_j ,  \rho_j ) \le  \sum_{i,j=0}^M D( \rho_i,  \rho_j ) = D(\rho, \rho).
    \]
    Finally, for the potential energy term, we have
    \begin{align} 
         \int_{\R^3} V \rho & = \int_{\R^3} \left( - \sum_{i=1}^M \dfrac{z_i}{| \bx - \bR_i |} \right) \left( \sum_{j=0}^M \rho_{j}(\bx) \right) \rd \bx  \nonumber \\
        &  = \sum_{j=1}^M - z_j \int_{\R^3} \dfrac{\rho_j(\bx)}{| \bx - \bR_j |}  \rd \bx + O(1), \label{eq:split_potential}
    \end{align}
    where, for the last equality, we expanded all the terms and used the fact that, for $i \neq j$, the support of $\rho_{j}$ avoids the singularity in $\bR_i$, so that the corresponding term is uniformly bounded by some $C N$.
\end{proof}

We localise the sequences $\widetilde{\gamma_{n}}$ and $ \widetilde{\bA_{n}}$ as in Lemma~\ref{lem:localisation}, and denote by 
$$\widetilde{\gamma_{j,n}} = \chi_j \widetilde{\gamma_n} \chi_j\quad \text{and}\quad\widetilde{\bA_{j,n}} := \chi_j' \widetilde{\bA_n}.$$
Since $\cE_\alpha(\widetilde{\gamma_n}, \widetilde{\bA_n})$ goes to $-\infty$, we have $\cE_\alpha(\widetilde{\gamma_n}, \widetilde{\bA_n}) + C < 0$ for $n$ large enough. So, by Lemma~\ref{lem:localisation},
\begin{align} 
    & \sum_{j=1}^M \left( \Tr \left( [\bsigma \cdot (\bp + \widetilde{\bA_{j,n}})]^2  \widetilde{\gamma_{j,n}}  \right)
    + D( \widetilde{\rho_{j,n}} , \widetilde{\rho_{j,n}} ) \right)
     + \dfrac{1}{8 \pi \alpha^2} \int_{\R^3} \widetilde{\bB_n}^2 \nonumber \\
     & \qquad \le
     \sum_{j=1}^M z_j \int_{\R^3} \dfrac{\widetilde{\rho_{j,n}}}{| \bx - \bR_j |}. \label{eq:bound_gammajn}
\end{align}
On the other hand, since the only negative term in the expression of $\cE_\alpha$ is the $\int_{\RR^3} V \rho$ term, we have $\int_{\RR^3 } V \widetilde{\rho_n} \to - \infty$. Together with~\eqref{eq:split_potential}, we get
\[
    \lambda_n := \sum_{j=1}^M z_j \int_{\R^3} \dfrac{\widetilde{\rho_{j,n}}(\bx)}{| \bx - \bR_j |} \rd \bx \xrightarrow[n \to \infty]{} +\infty.
\]
 \underline{Re-scaling.} We rescale the one-body density matrices and potential vectors with $\lambda_n \to \infty$, and we set, for $1 \le j \le M$,
\begin{align*}
    \gamma_{j,n}(\bx, \by) & := \dfrac{1}{\lambda_n^3} \widetilde{\gamma_{j,n}} \left(\dfrac{\bx}{\lambda_n}  + \bR_j, \dfrac{\by}{\lambda_n}  + \bR_j \right) \\ 
    \bA_{j,n} (\bx) & := \dfrac{1}{\lambda_n} \widetilde{\bA_{j,n}} \left( \dfrac{\bx}{\lambda_n} + \bR_j \right).
\end{align*}
The translations are chosen so that all $\gamma_{j,n}$ and $\bA_n$ are localised near $\bx = \bnull$, and the scaling is chosen so that 
\begin{equation} \label{eq:scaling_lambdan}
  \forall n \in \N, \quad  
  \sum_{j=1}^M z_j \int_{\R^3} \dfrac{\rho_{j,n}(\bx)}{ | \bx |} \rd \bx = 
  \dfrac{1}{\lambda_n} \sum_{j=1}^M z_j \int_{\R^3} \dfrac{\widetilde{\rho_{j,n}}(\by)}{ | \by - \bR_j |} \rd \by = 1.
\end{equation}
Our goal is to prove that each sequence $(\gamma_{j,n}, \bA_{j,n})$ converges, in an appropriate sense, to a zero-mode (potentially the null one). With the same computation as in~\eqref{eq:scaling1}-\eqref{eq:scaling2}, together with~\eqref{eq:bound_gammajn}, we get
\begin{align} \label{eq:bound_gammajn_1}
& \lambda_n \sum_{j=1}^M  \Tr \left( [\bsigma \cdot (\bp + \bA_{j,n})]^2  \gamma_{j,n}  \right)
+ \frac12 \sum_{j=1}^M D( \rho_{j,n} , \rho_{j,n} ) \nonumber \\
& \qquad \qquad \qquad + \dfrac{1}{8 \pi \alpha^2} \dfrac{1}{\lambda_n} \int_{\R^3} \widetilde{\bB_n}^2 
\le
\sum_{j=1}^M z_j \int_{\R^3} \dfrac{\rho_{j,n}(\bx)}{| \bx |}\rd \bx \quad (= 1).
\end{align}

\noindent \underline{Compactness.} We will use~\eqref{eq:bound_gammajn_1} in several ways. First, we notice that all terms in the left-hand side are positive, hence bounded. Let us prove that the functions $\bB_{j,n} := \curl \bA_{j,n}$ are bounded in $L^2(\R^3)$. Recall that $\widetilde{\bA_{j,n}} = \chi_j' \widetilde{\bA_n}$, so that $\widetilde{\bB_{j,n}} = \nabla \chi_j' \wedge \widetilde{\bA_n} + \chi_j' \widetilde{\bB_n}$.
We have by scaling and the inequality $(a+b)^2 \le 2 a^2 + 2b^2$ that
$$
 \int_{\R^3} | \bB_{j,n} |^2 := \dfrac{1}{\lambda_n} \int_{\R^3} | \widetilde{\bB_{j,n}} |^2
 \le \dfrac{2}{\lambda_n} \int_{\R^3} \left| \nabla \chi_j' \right|^2 | \widetilde{\bA_n} |^2 + 
 \dfrac{2}{\lambda_n} \int_{\R^3} \left| \chi_j' \right|^2 | \widetilde{\bB_n} |^2.
$$
Since $| \chi_j' | \in L^\infty$ and $\lambda_n^{-1} \| \widetilde{\bB_n}^2 \|_{L^2}$ is bounded by~\eqref{eq:bound_gammajn_1}, the second term is uniformly bounded. Similarly, $\left| \nabla \chi_j' \right|^2 $ is in $L^{3/2}(\R^3)$, while for $\bA\in H^1_{\div}$, we have with the Sobolev embedding $\| \bA \|_{L^6} \le C \| \nabla \bA \|_{L^2} = C \| \bB \|_{L^2}$. We deduce from Hölder's inequality that the first term is also uniformly bounded. So for all $0 \le j \le N$, $(\bB_{j,n})$ is bounded in $L^2(\R^3)$.

We now use the following Lemma, whose proof is postponed until the end of the section. In the sequel, we denote by $\fS_p$ the Schatten space of operators acting on $L^2(\R^3, \C^2)$: $\fS_1$ is the set of trace-class operators, and $\fS_2$ is the set of Hilbert-Schmidt operators.

\begin{lemma}[A compactness result]\label{lem:compactness}
 Let $(\gamma_n, \bA_n)$ be any sequence in $\cP\times H^1_{\div}$ such that
 $$
 \Tr\bra{ [\bsigma\cdot \bra{\bp+\bA_n} ]^2\gamma_n} + \| \bB_n \|_{L^2} + \Tr(\gamma_n) < C
 $$
Then there exists $(\gamma,\bA)\in \cP\times H^1_{\div}$ such that, up to a subsequence, we have the following convergences:
 \begin{itemize}
  \item $\bA_n\to \bA$ weakly in $L^6(\R^3)$ and strongly in $L^p_\loc(\R^3)$ for all $1\leq p< 6$;
  \item $\bB_n\to \bB := \curl \bA$ weakly in $L^2(\R^3)$;
  \item $\gamma_n \to \gamma$ weakly-* in $\fS_1$, and $(1 - \Delta)^{1/2} \gamma_n \to (1 - \Delta)^{1/2} \gamma$ weakly in $\fS_2$;
  \item $\rho_{\gamma_n}\to \rho_\gamma$ weakly in $L^1(\R^3) \cap L^3(\R^3)$ and strongly in $L^p_\loc(\R^3)$ for all $1 \le p < 3$;
  \item $\bsigma \cdot \bra{\bp+\bA_n}\gamma_n\to \bsigma \cdot \bra{\bp+\bA}\gamma $ weakly in $\fS_2$.
  \end{itemize}
In addition, we have the following inequalities:
\begin{align*}
    \int_{\R^3} \rho_\gamma & \le \liminf_{n \to \infty} \int_{\R^3} \rho_{\gamma_n}, \\
    \Tr\bra{ [\bsigma\cdot \bra{\bp+\bA}]^2\gamma} & \leq \liminf_{n\to +\infty} \Tr\bra{ [\bsigma\cdot \bra{\bp+\bA_n}]^2\gamma_n}, \\
    D(\rho_{\gamma},\rho_\gamma) & \leq \liminf_{n\to +\infty} D(\rho_{\gamma_n},\rho_{\gamma_n}), \\
    \int_{\RR^3} V\rho_\gamma & = \lim_{n\to +\infty} \int_{\RR^3} V\rho_{\gamma_n}, \\ \text{and } \int_{\RR^3}  | \cdot |^{-1}\rho_\gamma &= \int_{\RR^3} | \cdot |^{-1} \rho_{\gamma_n}.
\end{align*}

\end{lemma}

We use Lemma~\ref{lem:compactness} for each one of the sequences $(\gamma_{j,n},\bA_{j,n})$, and deduce that there exist states $(\gamma_j^*,\bA_j^*)$ satisfying the conditions of the lemma. In particular, we have $\bsigma \cdot \bra{p+\bA_{j,n}}\gamma_{j,n} \to \bsigma \cdot \bra{p+\bA_j^*}\gamma_j^*$, while, from~\eqref{eq:bound_gammajn_1} we have (we use that $\gamma_{j,n}^2 \le \gamma_{j,n}$) 
\begin{align*}
    \| \bsigma \cdot \bra{\bp+\bA_{j,n}}\gamma_{j,n} \|_{\fS_2}^2 & 
    = \Tr \left( \bsigma \cdot \bra{\bp+\bA_{j,n}}\gamma_{j,n}^2 \bsigma \cdot \bra{\bp+\bA_{j,n}} \right) \\
    & \le \Tr \left( [\bsigma \cdot \bra{\bp+\bA_{j,n}} ]^2 \gamma_{j,n} \right) \xrightarrow[n \to \infty]{} 0.
\end{align*}
We deduce that $\bsigma \cdot \bra{p+\bA_j^*}\gamma_j^* = 0$. In other words, for all $0 \le j \le M$, $(\gamma_j^*, \bA_j^*)$ is a zero-mode. However, we may have $\gamma_j^* = 0$. To prove that at least one of the limits is a non trivial zero mode, we use~\eqref{eq:scaling_lambdan} and pass to the limit $n \to \infty$. This gives
\begin{equation} \label{eq:limit_rhoj_over_x}
    \sum_{j=1}^N z_j \int_{\R^3} \dfrac{\rho_j^*(\bx)}{| \bx |} \rd \bx = 1.
\end{equation}
Hence, there is at least one $1 \le j \le M$ with $\rho_j^* \neq 0$, and the corresponding $(\gamma_j^*, \bA_j^*)$ is a non trivial zero-mode.

\medskip

\noindent \underline{Final inequality.}
It remains to use one of these non-trivial zero-modes to prove that $\alpha \ge \alpha_c(z,N)$. To do so, we use again~\eqref{eq:bound_gammajn_1}. We first need to split $\|  \widetilde{\bB_n} \|_{L^2}^2$ into a sum $ \sum_{j=1}^M \|  \bB_j \|_{L^2}^2$. This is possible at the limit. Indeed, by definition of $\bB_{j,n}$, we have
\begin{equation}\label{eq:Bjn}
\bB_{j,n}(\bx)=\frac{1}{\lambda_n^2}\bra{ \nabla \chi_j' \wedge \widetilde{\bA_{n}} + \chi_j' \widetilde{\bB_{n}} }\bra{\frac{\bx}{\lambda_n} + \bR_j}
\end{equation}
We already proved that both terms of~\eqref{eq:Bjn} are uniformly bounded in $L^2(\R^3)$. In addition, the support of the first term goes to infinity, hence the first term converge weakly to $0$ in $L^2(\R^3)$. Taking the weak-limit in~\eqref{eq:Bjn}, we obtain that $\lambda_n^{-2}  \chi_j'  \widetilde{\bB_{n}} \left( \lambda_n^{-1} \cdot + \bR_j\right)$ converges weakly to $\bB_j^*$ in $L^2(\R^3)$. In particular, 
\[
    \int_{\R^3} | \bB_j^* |^2 \le \liminf_{n \to \infty} \int_{\R^3}  \left( \dfrac{1}{\lambda_n^2}  \chi_j'  \widetilde{\bB_{n}} \left( \dfrac{\bx}{\lambda_n} +\bR_j \right)  \right)^2 \rd \bx
    = \liminf_{n \to \infty} \dfrac{1}{\lambda_n } \int_{\R^3} (\chi_j')^2| \widetilde{\bB_n} |^2.
\]
Summing these inequalities over $j$, and using the fact that $\sum_{j=1}^M (\chi_j')^2 \le 1$, we obtain
\begin{equation} \label{eq:sumBj^2}
    \sum_{j=1}^M \int_{\R^3} | \bB_j^* |^2 \le  \liminf_{n \to \infty} \dfrac{1}{\lambda_n } \int_{\R^3} | \widetilde{\bB_n} |^2.
\end{equation}
Using this together with~\eqref{eq:limit_rhoj_over_x} and Lemma~\ref{lem:compactness}, dropping the kinetic energy in~\eqref{eq:bound_gammajn_1} and passing to the limit, we obtain
\begin{equation} \label{eq:sum_zeromodes}
    \sum_{j=1}^M \left[ \frac12 D(\rho_j^*, \rho_j^*) + \dfrac{1}{8 \pi \alpha^2} \int_{\R^3} | \bB_j^* |^2 - z_j \int_{\R^3} \dfrac{\rho_j^*(\bx)}{| \bx | }\rd \bx \right] \le 0.
\end{equation}
By first restricting the sum in~\eqref{eq:sum_zeromodes} to all non-null zero-modes, then by noticing that if the sum of terms is negative, then at least one term is negative, we deduce that there is $1 \le j \le M$ such that $(\gamma_j^*, \bA_j^*)$ is a non-null zero-mode, with
\[
    D(\rho_j^*, \rho_j^*) + \dfrac{1}{8 \pi \alpha^2} \int_{\R^3} | \bB_j^* |^2 - z_j \int_{\R^3} \dfrac{\rho_j^*(\bx)}{| \bx | } \rd \bx \le 0.
\]
Since $\gamma_j^* \neq 0$ and $(\gamma_j^*, \bA^*)$ is a zero-mode, we must have $\bA_j^* \neq \bnull$ and $\bB_j^* \neq \bnull$. By performing a dilation (see~\eqref{eq:dilation}), we may always assume that $\| \bB_j^* \| = 1$.  Setting $\beta := \dfrac{-1}{8 \pi \alpha^2}$, this leads to
\[
    \beta \ge \frac12 D(\rho_j^*, \rho_j^*) - z_j \int_{\R^3} \dfrac{\rho_j^*(\bx)}{| \bx | } \rd \bx \quad \ge \beta_c(z_j, N) \ge \beta_c(z,N),
\]
where the last inequality come from the fact that $z \mapsto \beta_c(z,N)$ is non-increasing by Proposition~\ref{prop:betac}. This proves that $\alpha \ge \alpha_c(z,N)$ as claimed. 

\subsection*{Proof of Lemma~\ref{lem:compactness}}
It remains to prove Lemma~\ref{lem:compactness}. Let $(\gamma_n, \bA_n)$ be a sequence in $\cP\times H^1_{\div}$ such that
$$
\Tr\bra{ [\bsigma\cdot \bra{\bp+\bA_n} ]^2\gamma_n} + \| \bB_n \|_{L^2} + \Tr(\gamma_n) < C.
$$
\underline{Convergence of the magnetic fields.} As $(\bB_n)_{n\in\NN}$ is bounded, it converges, up to a subsequence, to some $\bB^*$ weakly in $L^2(\R^3)$. Since $\bA_n \in H^1_{\div}$, we get from the Sobolev embeddings that $ \| \bA_{n} \|_{L^6} \le C \| \nabla \bA_{n} \|_{L^2} = C \| \bB_{n} \|_{L^2}$, so the sequence $(\bA_{n})$ is bounded in $L^6(\RR^3)$, and converges to some $\bA^*$ weakly in $L^6(\RR^3)$. In distributional sense, we get that $\div \bA^* = 0$, and by identification of the limits, we also have $\curl \bA^* = \bB^* \in L^2(\R^3)$, hence $\bA^* \in H^1_{\div}$.

\medskip

\noindent \underline{Convergence for the density.} Let $\rho_n := \rho_{\gamma_n}$. Since $0 \le \gamma_n \le 1$ and $\Tr(\gamma_n)$ is bounded, we deduce that $(\rho_n)$ is bounded in $L^1(\R^3)$. Actually, we may bound $\rho_n$ in stronger topologies. First, for all $(\gamma, \bA) \in \cP \times H^1_{\div}$, we have
\begin{equation} \label{eq:expand_kinetic}
\Tr \left( \left[ \bsigma \cdot ( \bp + \bA) \right]^2  \gamma  \right) = \Tr \left( \left( \bp + \bA \right)^2 \gamma \right) + \int_{\R^3} \bB \cdot \bm_\gamma,
\end{equation}
where $\bm_\gamma (\bx) := \tr_{\C^2} \bra{\bsigma \cdot \gamma(\bx, \bx)}$ it the magnetisation, and satisfies $| \bm_\gamma |^2 = \rho_\gamma^2$. We recall the following classical inequalities

\begin{lemma} \label{lem:useful_ineq}
    For all $\gamma \in \cP$ and all $\bA \in H^1_{\div}$ such that $\Tr(\gamma) < \infty$, and $\Tr((\bp + \bA)^2 \gamma) < \infty$, the density $\rho_\gamma$ satisfies $\sqrt{\rho} \in H^1(\R^3)$, and the following inequalities hold:
    \[
    C_\LT \int_{\R^3} \rho_\gamma^{5/3} \le  \Tr(( \bp + \bA )^2 \gamma) \qquad \text{(Lieb-Thirring)}
    \]
    and
    \[
    \int_{\R^3} | \nabla \sqrt{\rho_\gamma} |^2 \le \Tr( (\bp + \bA )^2 \gamma) \qquad \text{(Hoffman-Ostenhof)}.
    \]
    In particular, together with Sobolev inequality, there is $C_2 \in \R^+$ such that
    \[
       C_2  \left( \int_{\R^3} \rho_{\gamma}^{2} \right)^{2/3}\le  \Tr((\bp + \bA )^2 {\gamma}).
    \]
\end{lemma}
The proof can be found for instance in~\cite[Lemma 4.3 and Lemma 8.4]{stability} for the two first inequalities. The last one is obtained by Hölder's inequality.

From~\eqref{eq:expand_kinetic} and Lemma~\ref{lem:useful_ineq} together with Cauchy-Schwarz inequality, we deduce that
\[
    \Tr \left( \left[ \bsigma \cdot ( \bp + \bA) \right]^2  \gamma  \right) \ge C_2 \left( \int_{\R^3} \rho_\gamma^2 \right)^{2/3} - \| \bB \|_{L^2} \left( \int_{\R^3} \rho_\gamma^2 \right)^{1/2}.
\]
In particular, since the sequence $\| \bB_{n} \|_{L^2}^2$ is bounded, and the fact that $C_2 X^{2/3} - B X^{1/2}$ goes to infinity as $X \to \infty$, we deduce that $(\rho_{n})$ is bounded in $L^2(\R^3)$. 
Using again~\eqref{eq:expand_kinetic} and Lemma~\ref{lem:useful_ineq}, we also have
\[
     \Tr \left( \left[ \bsigma \cdot ( \bp + \bA_n) \right]^2  \gamma_{n}  \right) \ge \int_{\R^3} | \nabla \sqrt{\rho_{n}} |^2 - \| \bB_n \|_{L^2} \left( \int_{\R^3} \rho_{n}^2 \right)^{1/2},
\]
and $(\nabla \sqrt{\rho_{n}})$ is also bounded in $L^2(\R^3)$. From the Sobolev embedding, this implies that, up to a subsequence, there is $\rho^* \in L^1 \cap L^3$ such that $\rho_n$ converges weakly to $\rho^*$ in $L^1(\R^3) \cap L^3(\R^3)$, and strongly in $L^p_\loc(\R^d)$ for all $1 \le p < 3$.
\begin{remark} \label{rem:rho_zeromode}
    The same argument shows that if $(\gamma, \bA)$ is a zero-mode, then $\rho_\gamma \in L^1(\R^3) \cap L^3(\R^3)$ and $\sqrt{\rho_\gamma} \in H^1(\R^3)$.
\end{remark}

\noindent \underline{Convergence for the one-body density matrices.}
Since $\gamma_n$ is bounded in norm ($0 \le \gamma_n \le 1$), and $\| \gamma_n \|_{\fS_1} = \Tr(\gamma_n) \le C$ is also bounded, we deduce that $\gamma_n$ is bounded in all Schatten spaces $\fS_p$. In particular, up to a subsequence, it converges to some $\gamma^* \in \fS_1$ for the weak operator topology:
\begin{equation} \label{eq:WOT}
    \forall f,g \in L^2(\R^3, \C^2), \quad  \lim_{n \to \infty} \langle f, \gamma_n g \rangle = \langle f, \gamma^* g \rangle.
\end{equation}
Using that $- \Delta = (\bp + \bA - \bA)^2 \le 2 (\bp + \bA)^2 + 2 \bA^2$ together with~\eqref{eq:expand_kinetic} and the Hölder's inequality, we also have
\begin{align*}
    \Tr (- \Delta \gamma_n) & \le 2 \Tr ( (\bp + \bA_n)^2 \gamma_n) + 2 \int_{\R^3} \bA_n^2 \rho_n \\
    & \le 2 \Tr \left( \left[ \bsigma \cdot ( \bp + \bA_n) \right]^2  \gamma_{n}  \right) + 2 \| \bB_n \|_{L^2} \| \rho_n \|_{L^2} + 2 \| \bA_n \|_{L^6}^2 \| \rho_n \|_{L^{3/2}},
\end{align*}
so the sequence $(-\Delta \gamma_n)_{n \in \N}$ is bounded in $\fS_1$. Up to a subsequence, it converges for the weak operator topology to an operator $T$. To prove that $ T = -\Delta \gamma^*$ we consider $f,g \in C^\infty_0(\R^3, \C^2)$ and using~\eqref{eq:WOT} and the fact that $(1 - \Delta) f \in L^2(\R^3)$, we obtain
\begin{align} 
    \langle f,(1+ T) g \rangle & = \lim_{n \to \infty}  \langle f, (1 - \Delta) \gamma_n g \rangle 
    = \lim_{n \to \infty}  \langle (1 - \Delta) f, \gamma_n g \rangle  \nonumber\\
    & = \langle (1 - \Delta) f, \gamma^* g \rangle
    = \langle  f, (1 - \Delta) \gamma^* g \rangle.\label{eq:semiContinuity_WOT}
\end{align}

We now prove that 
$\rho^* = \rho_{\gamma^*}$. For any smooth function $\chi \in C^\infty_0(\R^3)$, we have that the operator $\chi (1 - \Delta)^{-1}$ is compact by the Kato-Simon-Seiler inequality~\cite{trace_ideals}. Together with the fact that $(1 - \Delta) \gamma_n$ converges to $(1 - \Delta) \gamma^*$ for the weak operator topology, we obtain
\begin{align*}
    \int_{\R^3} \chi \rho^* & = \lim_{n \to \infty} \int_{\RR^3} \chi \rho_n = \lim_{n \to \infty} \Tr\bra{\chi (1-\Delta) ^{-1/2} (1-\Delta) ^{1/2}\gamma_n} \\
    & = \Tr\bra{\chi (1-\Delta) ^{-1/2} (1-\Delta) ^{1/2}\gamma} = \int_{\R^3} \chi \rho_{\gamma^*}.
\end{align*}
By identification, this implies that $R^* = R_{\gamma^*}$.

\medskip

\noindent \underline{Convergence of $\bsigma \cdot (\bp + \bA_n) \gamma_n$.} We already proved that the sequence $\bsigma \cdot \bp \gamma_n$ is bounded in $\fS_2$, and converges, up to a subsequence, to $\bsigma \cdot \bp \gamma^* \in \fS_2$ (see~\eqref{eq:semiContinuity_WOT}). In addition, the sequence $\bsigma \cdot \bA_n \gamma_n$ is also bounded in $\fS_2$, hence converges up to a subsequence to some $T \in \fS_2$ for the weak operator topology. To prove that $T = \bsigma \cdot \bA^*\gamma^*$, we take $f, g \in C^\infty_0(\R^3, \C^2)$ and get
\begin{align*}
    \langle f, T g \rangle & = \lim_{n \to \infty} \langle f, \bsigma \cdot \bA_n \gamma_n g \rangle 
    = \lim_{n \to \infty} \int_{\R^3} \left( \bsigma \cdot \bA_n f \right) R_{n} g \\
    & =   \int_{\R^3} \left( \bsigma \cdot \bA^* f \right) R_{\gamma^*} g = \langle f, \bsigma \cdot \bA^* \gamma^* g \rangle,
\end{align*}
where we used the strong convergence of both $R_n$ and $\bA_n$ in $L^2_\loc$, together with the fact that $f$ and $g$ are smooth and compactly supported. Altogether, this proves that $\bsigma \cdot (\bp + \bA_n) \gamma_n$ converges up to a subsequence to $\bsigma \cdot (\bp + \bA^*) \gamma^*$ weakly in $\fS_2$.

\medskip

\noindent \underline{Last inequalities.} Finally, the three first inequalities of Lemma~\ref{lem:compactness} come from the lower semi-continuity of the functionals. For the last one, we have for $A > 0$,
\begin{equation*} 
\left| \int_{\RR^3} \dfrac{(\rho_n - \rho^*)(\bx)}{| \bx - \bR |} \rd \bx \right| = \left| \int_{\cB(0, A)} \dfrac{(\rho_n - \rho^*)(\bx + \bR) }{| \bx |} \rd \bx \right| + 
 \left| \int_{\cB(0, A)^c} \dfrac{(\rho_n - \rho^*)(\bx + \bR)}{| \bx |} \rd \bx \right|.
\end{equation*}
The second term is uniformly bounded by $2N A^{-1}$, hence is arbitrary small as $A \to \infty$. Then, since $\1(| \bx | < A) | \bx |^{-1}$ is in $L^2(\R^3)$ and compactly supported, together with the fact that $\rho_n \to \rho^*$ strongly in $L^2_\loc(\R^3)$, the first term goes to $0$ as $n \to \infty$. This proves that $\int_{\R^3} \rho_n  | \cdot - \bR |^{-1}$ converges to $\int_{\R^3} \rho^*  | \cdot - \bR |^{-1}$, and the proof follows.

\section{Proof of Theorem~\ref{th:existence_minimiser_per} for periodic systems}
\label{sec:proof:(in)stablity_per}
The proof for periodic systems is very similar to the previous ones, thanks to our cut-off functions. Let us highlight some differences. The first main difference is that all cut-off functions need to be periodised, so we set
\[
    \chi_\per = \sum_{\bR \in \Lat} \chi (\cdot - \bR) \quad \text{and} \quad 
    \chi'_\per = \sum_{\bR \in \Lat} \chi'(\cdot - \bR),
\]
where $\chi$ and $\chi'$ are the cut-off functions considered in~\eqref{eq:def:chichi'}. Another difference is that we need to enforce the condition $\VTr(\gamma_\lambda) = N$ for all $\lambda$. To do so, we need to add some electrons far from the singularities. In practice, we take $\gamma_0 \in \cP_{\per}^N$ a smooth one-body matrix with $N$ electrons per unit cell, and so that $\1_{\WS} \gamma_0 \1_{\WS}$ has support away from $\bigcup_j \cB(\bR_j, r/2)$, and we consider test functions of the form
\[
    \gamma_\lambda + \eta_\lambda \gamma_0, \quad \text{with} \quad \eta_\lambda := 1 - \frac1N \Tr( \gamma_\lambda) \in [0,1].
\]
Since $\gamma_0$ is smooth, it only adds a bounded contribution to the total energy. We leave the details for the sake of brevity.
The final difference is that the full space Coulomb kernel $| \bx |^{-1}$ is replaced with its periodic counterpart $G_\Lat$. However, $G_\Lat$ and $| \cdot |^{-1}$ have the same singularity as $\bx \to \bnull$, since
\[
    F(\bx) := G_{\Lat}(\bx) - \dfrac{1}{ | \bx |}
\]
satisfies $\Delta F = 0$ on $\WS$, hence is smooth and bounded on $\WS$.

\section{Proof of Proposition~\ref{prop:betac}: properties of $\beta_c$} 
\label{sec:proof:betac}

In this section, we study the properties of $\beta_c(z,N)$ defined in~\eqref{eq:def:betazN}, and prove Proposition~\ref{prop:betac}. We denote by
\[
    \cF_z(\gamma, \bA) := \left( \frac12 D(\rho_\gamma, \rho_\gamma) - z \int_{\R^3} \dfrac{\rho_\gamma(\bx)}{| \bx |} \right) / \| \bB \|_{L^2}^2.
\]
This functional is invariant by the scaling~\eqref{eq:scaling1}-\eqref{eq:scaling2}, so we can always choose $\| \bB \|_{L^2} = 1$, in which case we recognise the functional that is minimised in the definition of $\beta_c(z,N)$. 

We first prove that $\beta_c(z,N) < 0$ for all $z,N > 0$. From~\cite{LY3}, there exist $\Psi \in L^2(\R^3)$ and $\bA \in H^1_{\div}$ such that $\| \Psi \|_{L^2} = 1$, $\| \bB \|^2 = 1$ and $\Psi \in \Ker( \bsigma \cdot (\bp + \bA))$. In particular, for all $0 < \varepsilon < 1$, the one-body density matrix $\gamma_\varepsilon := \varepsilon | \Psi \rangle \langle \Psi |$ satisfies $0 \le \gamma_\varepsilon \le 1$ and $\Ran \gamma_{\varepsilon} \subset \Ker( \bsigma \cdot (\bp + \bA))$. So $(\gamma_\varepsilon, \bA) \in \cZ$, and
\[
     \cF_z(\gamma_\varepsilon, \bA)  = \frac{\varepsilon^2}{2} D(| \Psi |^2, | \Psi |^2) - z \varepsilon \int \dfrac{| \Psi |^2(\bx)}{| \bx |} \rd \bx.
\]
This term becomes strictly negative for $\varepsilon$ small enough, hence $\beta_c(z,N) < 0$.

\medskip

Since the map $z \mapsto \cF_z$ is pointwise non-increasing, then so is $z \mapsto \beta(z,N)$. In addition, as $\beta(z,N)$ is defined as a minimisation problem of the same function $\cF_z$, but on a set which increases with $N$, the map $N \mapsto \beta(z,N)$ is non-increasing.
Finally, for any zero-mode $(\gamma, \bA) \in \cZ$ with $\| \bB \|_{L^2} = 1$, we have for all $\lambda > 0$ that $\cF_z(\gamma, \bA) = \cF_z^\lambda(\gamma, \bA)$, where we set
\[
    \cF_z^\lambda(\gamma, \bA) := \frac12 D(\rho, \rho)- z \int_{\R^3} \dfrac{\rho(\bx)}{| \bx |} \rd \bx + \lambda \Tr \left( \left[ \bsigma \cdot ( \bp + \bA) \right]^2 \gamma \right),
\]
which is a penalised version of $\cF_z$. We deduce that
\[
    \beta_c(z,N) \ge \beta_c^\lambda(z) := \inf \{  \cF_z^\lambda(\gamma, \bA), \ \gamma \in \cP, \ \bA \in H^1_{\div}, \ \| \bB \|_{L^2} = 1  \}.
\]
This is a minimisation problem without the constraints that $(\gamma, \bA)$ is a zero-mode, nor that $\Tr(\gamma) \le N$. We now bound from below $\beta_c^\lambda(z)$, and optimise the result in~$\lambda$. Using Lemma~\ref{lem:useful_ineq}, and the fact that $\| \bB \|_{L^2} = 1$, we get

\begin{align*}
\cF_z^\lambda(\gamma, \bA) \ge \frac12 D(\rho, \rho) -  z \int_{\R^3} \dfrac{\rho(\bx)}{| \bx |} \rd \bx + \lambda \left( \frac{C_{\rm LT}}{2} \int_{\R^3} \rho^{5/3} + \frac{1}{2}\int_{\R^3} | \nabla \sqrt{\rho} |^2 - \| \rho \|_{L^2} \right).
\end{align*}
The right-hand side depends only on the density $\rho$, so we can optimise it over the set $ R := \left\{ \rho\in L^1(\RR^3),\;\sqrt{\rho}\in H^1(\RR^3) , \ \rho \ge 0 \right\}$. Using the Sobolev embedding $H^1(\R^3) \hookrightarrow L^6(\R^3)$ and the Young's inequality $ab \le \varepsilon a^{8/3} + \frac58 (\frac{3}{8 \varepsilon})^{3/5} b^{8/5}$, we obtain
\begin{align*}
 \norm{\rho}_{L^2}&\leq \norm{\rho}_{L^{5/3}}^{5/8} \norm{\rho}_{L^{3}}^{3/8}
 \leq  C \norm{\rho}_{L^{5/3}}^{5/8} \norm{\nabla\sqrt{\rho}}_{L^2}^{3/4} \\
 & \leq C \varepsilon \norm{\rho}_{L^{5/3}}^{5/3} + C \frac53 \left( \frac{3}{8 \varepsilon} \right)^{3/5} \norm{\nabla\sqrt{\rho}}_{L^2}^{6/5}.
\end{align*}
We choose $\varepsilon$ so that $C \varepsilon = C_{\rm LT}/4$. Since the function $Y \mapsto Y^2 - C Y^{6/5}$ is bounded from below, where $Y := \| \nabla \sqrt{\rho} \|_{L^2}$, we deduce that there is a constant $C$ large enough so that
\[
 \frac{C_{\rm LT}}{2} \| \rho \|_{L^{5/3}}^{5/3} + \frac{1}{2} \| \nabla \sqrt{\rho} \|_{L^2}^2 - \| \rho \|_{L^2} \ge 
 \frac{C_{\rm LT}}{4} \| \rho \|_{L^{5/3}}^{5/3} -C.
\]
Altogether, we proved that
\begin{align*}
    \beta_c (z) := \inf_N \beta(z,N) & \ge \inf_{\rho\in R} \left\{ \frac{\lambda C_{\rm LT}}{4} \int_{\R^3} \rho^{5/3}  + \frac12 D(\rho, \rho) -  z \int_{\R^3} \dfrac{\rho(\bx)}{| \bx |} \rd \bx \right\} - C \lambda .
 \end{align*}
Setting  $\rho_0(\bx) := a^{-1} b^{-3} \rho\bra{ \bx/b}$ for some constants $a,b >0$ that we specify later, we get
\begin{align*}
\beta_c (z) & \ge \inf_{\rho_0\in R} \left\{ \frac{\lambda C_{\rm LT}}{4} a^{5/3} b^2 \int_{\R^3} \rho_0^{5/3}  + \frac{a^2 b}{2} D(\rho_0, \rho_0) -  z a b\int_{\R^3} \dfrac{\rho_0(\bx)}{| \bx |} \rd \bx \right\} - C \lambda .
\end{align*}
We choose $a = z$, $b =z^{-5/6}$, and $\lambda = 4/C_{\rm LT} z^{7/6}$, and get that $$\beta(z,N)  \ge (I^{\rm TF} -         C) z^{7/6},$$ where $I^{\rm TF}$ is the Thomas-Fermi energy
\[
   I^{\rm TF} := \inf_{\rho_0\in L^1\cap L^{5/3},\atop \rho_0\geq 0} \left\{\int_{\R^3} \rho_0^{5/3}  + \frac12 D(\rho_0, \rho_0) -  \int_{\R^3} \dfrac{\rho_0(\bx)}{| \bx |} \rd \bx \right\}.
\]
The fact that $I^{\rm TF}$ is well defined and finite is a result by Lieb and Simon~\cite[Theorem II.3]{Lieb1977} based on Teller's lemma~\cite{Teller}.

Finally, for any $z$ and $N$, the problem defining $\beta(z,N)$ has minimisers. Indeed, any minimising sequence satisfies the conditions of Lemma~\ref{lem:compactness}. By the results of the Lemma, the limit is a zero-mode that minimises $\cF_z$. As $\beta(z,N)<0$,  the limit is a non trivial zero mode, we can choose $\norm{\bB}=1$ by scaling without loss of generality.


\section{Proof of Theorems~\ref{th:existence_minimiser} and \ref{th:existence_minimiser_per}}
\label{sec:proof_existence}

\subsection*{Existence of minimisers}

We now prove that when $N \le Z$, that is when the system is neutral or positively charged, then the finite problem~\eqref{eq:def:Ialpha} have minimisers whenever $\alpha < \alpha_c(z,N)$. In the periodic setting, we always have $N = Z$, and minimisers always exist when $\alpha < \alpha_c(z,N)$. We detail the proof for the molecular case, and later explain the modifications for the periodic case.

Let $\alpha < \alpha_c(z,N)$, and let $(\gamma_n, \bA_n) \in \cP^N \times H^1_{\div}$ be a minimising sequence for $I_\alpha$. We first prove that $\| \bB_n \|_{L^2}$ is uniformly bounded. To do so, we introduce an intermediate $\alpha < \alpha' < \alpha_c(z,N)$, and notice that
\begin{align*}
    \cE_{\alpha}(\gamma_n, \bA_n) & = \cE_{\alpha'}(\gamma_n, \bA_n) + \dfrac{1}{8 \pi} \left(\dfrac{1}{\alpha^2} - \dfrac{1}{(\alpha')^2} \right) \int_{\RR^3} \bB_n^2 \\
    & \ge I(\alpha') + \dfrac{1}{8 \pi} \left(\dfrac{1}{\alpha^2} - \dfrac{1}{(\alpha')^2} \right) \int_{\RR^3} \bB_n^2,
\end{align*}
so the sequence $(\bB_n)$ is bounded in $L^2(\R^3)$. Using estimates similar to the ones of Lemma~\ref{lem:useful_ineq}, together with the fact that, by Hölder's inequality, for $R > 0$ large enough so that $\cB(\bnull, R)$ contains all the nuclei,
\begin{align*}
    \left| \int_{\R^3} V \rho_n \right| & = \int_{\cB(\bnull, R)} V \rho_n + \int_{\cB(\bnull, R)^c} V \rho_n \\
    & \le \| V \1(| \bx | \le R) \|_{L^{5/2}} \| \rho_n \|_{L^{5/3}} + \|  V \1(| \bx | > R) \|_{L^\infty} N,
\end{align*}
we deduce that $\Tr \left( [ \bsigma \cdot (\bp + \bA) ]^2 \gamma_n \right)$ is uniformly bounded. We can therefore apply Lemma~\ref{lem:compactness}, and deduce that, up to a subsequence, there is $(\gamma^*,\bA^*)\in \cP \times H^1_{\div}$ such that
$$
\cE_\alpha(\gamma^*, \bA^*)\leq \liminf \cE_\alpha(\gamma_n, \bA_n)
\quad \text{and} \quad \Tr\bra{\gamma^*}\leq N. 
$$
So $(\gamma^*, \bA^*)$ is a minimiser for $\cE_\alpha$ if and only if $\Tr(\gamma^*) = N$. Unfortunately, we only have that $\Tr(\gamma^*) \le N$ at this point, and electrons may leak to infinity. The existence of a minimiser is a consequence of the following Lemma.
\begin{lemma}
    If $N \le Z$, then $\Tr(\gamma^*) = N$.
\end{lemma}
The proof of this lemma was done in~\cite[Theorem II.13]{carlos}. It relies on the fact that, if $\Tr(\gamma^*) < Z$, then the corresponding mean-field Hamiltonian $H_{\rho^*, \bA^*}$ defined in~\eqref{eq:Euler-Lagrange} has an infinity of negative eigenvalues.

\begin{remark}
    In the periodic setting, electrons cannot leak away. Indeed, from the boundedness of $\sqrt{\rho_n}$ in $L^2_\per$, we deduce that $\rho_n$ converges strongly to $\rho^*$ in $L^p_\per$ for all $1 \le p < 3$. In particular, 
    \[
    \VTr_\Lat(\gamma^*)= \int_{\WS} \rho^* = \lim_{n \to \infty} \int_{\WS} \rho_n = \lim_{n \to \infty} \VTr_\Lat(\gamma_n) = N.
    \]
\end{remark}

\subsection*{The Euler-Lagrange equations} 
Finally, we derive the Euler-Lagrange equations. Let $(\gamma^*, \bA^*) \in \cP^N \times H^1_{\div}$ be a minimiser (with $N$ not necessarily smaller than $Z$) of $\cE_\alpha$. Then for all $\gamma \in \cP^N$, we have
\[
   \forall 0 \le t \le 1, \quad  \cE_\alpha \left( (1 - t) \gamma^* + t \gamma \right) \ge \cE_\alpha (\gamma^*).
\]
Hence the derivative at $t = 0^+$ must be positive, that is $ \Tr \left( H_{\rho^*, \bA^*}(\gamma - \gamma^*) \right) \ge 0$, with
\[
   H_{\rho^*, \bA^*} := \frac12 \left[  \bsigma \cdot (\bp + \bA^*) \right]^2 + V  + \rho^* * | \cdot |^{-1}.
\]
Since this is valid for all $\gamma \in \cP^N$, we get that $\gamma^*$ is also the minimiser of the linearised problem
\begin{equation} \label{eq:linearised_minimisation}
    \gamma^* \in {\rm argmin} \left\{ \Tr \left( H_{\rho^*, \bA^*} \gamma \right), \ \gamma \in \cP^N   \right\}.
\end{equation}
It is proved in~\cite[Appendix A]{erdos1995magnetic} that the form domain of $H_{\rho^*, \bA^*}$ is $H^1(\R^3, \C^2)$, and that its essential spectrum is $[0, \infty)$. In particular, $\gamma^*$ is of the form 
\[
    \gamma^* = \sum_{k \ge 0} n_k | \phi_k \rangle \langle \phi_k |, \quad \text{where} \quad  0 \le n_k \le 1 \quad \text{and} \quad \sum_{k \ge 0} n_k = N,
\] 
and where the functions $\phi_k \in H^1(\R^3, \C^2)$ are orthonormal in $L^2(\R^3, \C^2)$, and satisfy $H_{\rho^*, \bA^*} \phi_k = - \lambda_k \phi_k$ with $\lambda_k \ge 0$. So $ - \left( - \frac12 \Delta + \lambda_k \right) \phi_k = f$, with
\[
    f := \bA^* \cdot (- \ri \nabla \phi_k) + \left( \frac12 | \bA^* |^2 + \frac12 \sigma \cdot \bB^* + V + \rho^**| \cdot |^{-1} \right) \phi_k.
\]
We now use classical elliptic arguments to prove that $\phi_k \in H^2(\R^3, \C^2)$. In the sequel, we write $L^p$ for $L^p(\R^3, \C^2)$ for shortness. Let us first prove that $f \in L^{3/2}$.
Since $\bA^* \in H^1_{\div}$, while $\phi_k \in H^1$, we have $\bA \cdot (- \ri \nabla \phi_k) \in L^{3/2}$, $ | \bA^* |^2 \phi_k \in L^{6/5} \cap L^2$ and $\sigma \cdot \bB^* \phi_k \in L^1 \cap L^{3/2}$. Also, since $V \in L^{3 - \varepsilon} + L^{3+ \varepsilon}$, we deduce that $V \phi_k \in L^{6/5 + \varepsilon} \cap L^{2 - \varepsilon}$. Finally, since $\rho^* \in L^1 \cap L^3$, we have $\rho^* *| \cdot |^{-1} \in L^{3 + \varepsilon} \cap L^{\infty}$, and therefore  $\rho^* *| \cdot |^{-1} \phi_k \in L^{6/5} \cap L^6$. Altogether, we obtain $f \in L^{3/2}$.
We write 
$$(- \Delta + 1) \phi_k = f_{3/2} + (1 - \lambda_k) \phi_k \in L^{3/2} + L^2.$$ 
We take the convolution with the Yukawa kernel 
$$Y(x) := \frac{\re^{ - \sqrt{1 + \lambda_k} \av{x}}}{\av{x}},$$ 
which belongs to $L^1 \cap L^{3}_w$, and deduce from Young's inequality that $\phi_k \in L^2 \cap L^\infty$ (this already proves that the density is bounded). Also, since $\nabla Y \in L^1 \cap L^{3/2}_w$, we have $\nabla \phi_k \in L^{3/2} \cap L^3+ L^2\cap L^6\subset L^2\cap L^3$. 

This allows to improve regularity for $f$. We have $\bA^* (- \ri \nabla \phi_k)  \in L^2$, $\sigma \cdot \bB^* \phi_k \in L^2$, and $V \phi_k \in L^2$, so $f \in L^2$. This gives $(- \Delta + 1) \phi_k \in L^2$, and finally $\phi_k \in H^2(\R^3)$ by usual elliptic regularity.

\medskip

We now focus on the Euler-Lagrange equation for $\bA^*$. Using the fact that the condition $\div(\bA) = 0$ is equivalent to $\| \div(\bA) \|_{L^2}^2 = 0$, we obtain that $\bA^* \in H^1_{\div}$ satisfies
\begin{equation}\label{eq:EulerLagrange_A2}
 \frac12 ( \bj^* + \curl \bm^* ) + \bA^* \rho + \frac{1}{4 \pi \alpha} (- \Delta \bA^*) = 0.
\end{equation}
Now, since $\phi_k \in H^2(\R^3, \C^2)$, we have $\bj^* := 2 \Im \sum_{k \ge 0} {\rm Tr}_{\C^2} \bra{\overline{\phi_k} \nabla \phi_k} \in L^6(\R^3, \R^3)$ and similarly $\curl \bm^* \in L^6(\R^3, \R^3)$. Together with the fact that $\bA^* \rho \in L^6(\R^3, \R^3)$, we deduce that $-\Delta \bA^* \in L^6(\R^3, \R^3)$. Hence $( - \Delta + 1) \bA^* \in L^6(\R^3, \R^3)$, and we deduce that (see for instance~\cite{gilbarg2015elliptic}) $\bA^* \in W^{2, 6}(\R^3, \R^3)$.

\medskip

The computations are similar in the periodic case.

\printbibliography 	

\end{document}